\documentclass[submission,copyright,creativecommons]{eptcs}
\providecommand{\event}{EPTCS 2011} 
\usepackage{breakurl}             

\usepackage{mytheorems}
\usepackage{amsmath}
\usepackage{mathrsfs}
\usepackage{amssymb}
\usepackage{latexsym} 
\usepackage{verbatim}
\usepackage{graphicx}
\usepackage{float}
\usepackage{url}
\usepackage{enumerate}
\usepackage{euscript}
\usepackage{subfigure}
\usepackage{floatflt}
\usepackage[all]{xy}

\title{Lazy AC-Pattern Matching for Rewriting}

\author{Walid Belkhir and Alain Giorgetti\\
 FEMTO-ST, University of Franche-Comt\'{e},\\
  16 route de Gray, 25030 Besan\c{c}on cedex, France\\
INRIA Nancy - Grand Est, CASSIS project, 54600 Villers-l\`es-Nancy,
France\\
\url{{walid.belkhir,alain.giorgetti}@femto-st.fr}\\
}
 \def\titlerunning{Lazy AC-matching for rewriting}
 \def\authorrunning{W. Belkhir \& A.Giorgetti}

\newcommand \matchth[3]{#1 {\ll}_{{ }_{#3}} {#2}}
\newcommand \matchAC[2]{\matchth{#1}{#2}{AC}}
\newcommand \rrule[2]{#1\rightarrow #2}

\newcommand \up {\uparrow}
\newcommand \set[1]{\{#1\}}

\newcommand \varset{\ensuremath{\mathcal{X}}}

\newcommand \trip[1]{\langle #1 \rangle}
\newcommand \lazy{$\textit{Lazy}^{\downarrow}\,$}
\newcommand \Ronetwo{\ensuremath{\mathcal{R}_1\cup\mathcal{R}_2}} 
\newcommand \Rthree{\ensuremath{\mathcal{R}_3^{\downarrow}}}

\newcommand \mycomment[1]{}

\newcommand \llist[1]{\ensuremath{\textit{LList}(#1)}}
\newcommand \Eu[1]{\EuScript{#1}}

\newcommand \numberAC[1]{\sharp_{{}_{AC}}{(#1)}}

\numberwithin{subcase}{case}

\begin{document}
\maketitle

\setlength{\abovecaptionskip}{-0.5cm}
\setlength{\belowcaptionskip}{-0.2cm}

\begin{abstract}
We define a lazy pattern-matching  mechanism  modulo associativity and commutativity. 
The solutions of a pattern-matching problem are stored in a lazy list composed of 
a first substitution at the head and  a non-evaluated object  that encodes  the
remaining  computations. We integrate the lazy AC-matching in a strategy language: 
 rewriting rule and strategy application produce a lazy list of terms.  
\end{abstract}

\section{Introduction}\label{Introduction:Sec}
Term rewriting modulo associativity and commutativity of some function symbols,
known as AC-rewri\-ting, is a key operation in many programming languages,
theorem provers and computer algebra systems. Examples of AC-symbols are $+$ and
$*$ for addition and multiplication in arithmetical expressions, $\lor$ and
$\land$ for disjunction and conjunction in Boolean expressions, etc.
AC-rewriting performance mainly relies on that of its AC-matching algorithm. On
the one hand, the problem of deciding whether an AC-matching problem has a
solution is NP-complete \cite{Benanav:1985:CMP}. On the other hand, the number of
solutions to a given AC-matching problem can be exponential in the size of its
pattern. Thus many works propose optimizations for AC-matching. One can divide
optimized algorithms in two classes, depending on what they are designed for. In
the first class  some structural properties are imposed on the terms, and the
pattern falls into one of several forms for which efficient algorithms can be
designed. Examples are the depth-bounded patterns in the many-to-one matching
algorithm used by Elan~\cite{Kirchner:2001:Promoting:AC} and greedy matching
techniques adopted in Maude~\cite{MaudeBook07}. In the second class there is no
restriction on the structural properties of the terms. Algorithms in this class
are search-based, and use several techniques to collapse the search space, such
as constraint propagation on non linear variables \cite{gramlich-unif88},
recursive decomposition via bipartite graphs \cite{Bipartie-Eker95}, ordering
matching subproblems based on constraint propagation \cite{Fast-Eker96}
 and Diophantine techniques \cite{Single-Eker02}.

Formal semantics proposed so far for AC-rewriting enumerate all the possible
solutions of each matching problem. More precisely, the application modulo AC of
a rewrite rule $\rrule{l}{r}$ to a given term $t$ usually proceeds in two steps.
Firstly, all the solutions (i.e. substitutions) $\sigma_1$, \ldots, $\sigma_n$
($n \ge 0$) of the AC-matching problem whether the term $t$ matches the pattern
$l$ are computed and stored in a structure, say a set
$\set{\sigma_1,\ldots,\sigma_n}$. Secondly, this set is applied to $r$ and the
result is the set $\set{\sigma_1(r),\ldots, \sigma_n(r)}$. Other structures such
as multisets or lists can alternatively be used for other applications of the
calculus. Directly implementing this \emph{eager} semantics is clearly less
efficient than a lazy mechanism that only computes a first result of the
application of a rewrite rule and allows the computation by need of the remaining
results.  As far
as we know no work defines the AC-matching in a lazy way and integrates it in a
rewriting semantics.

Another motivation of this work lies in our involvement in the formulation of
the homogenization of partial derivative equations within a
 symbolic computation tool~\cite{EuroSim11,CFM11}.
For this purpose, we have designed and developed a rule-based language
 called \texttt{symbtrans} \cite{BGL-JSC10} for ``symbolic transformations''
 built on the computer algebra system Maple. Maple  pattern-matching procedures
 are not efficient and its rewriting kernel is very elementary. Besides, Maple
 is a strict language, it does not provide any laziness feature.
We plan to extend \texttt{symbtrans} with AC-matching.

In this paper we first specify a lazy AC-matching algorithm which computes
 the solutions of an AC-matching problem by need.
 Then we integrate the lazy AC-matching in a strategy
language. In other words we define a lazy semantics for rule and strategy
functional application. Our goal is to specify the lazy AC-matching and strategy
 semantics towards an implementation in a strict language, such as
Maple. We reach this goal by  representing  lazy lists by means of explicit
objects.

The paper is organized as follows. Section~\ref{Notations:Sec} introduces some
terminology and notations. Section~\ref{surj:sec} shows a connection between
AC-matching and surjective functions, used in the remainder of the paper.
Section~\ref{ac:match:sec} formally defines  a lazy semantics of AC-matching
(with rewrite rules). It states its main properties and shows how it differs
from an eager semantics.
Section~\ref{lazy:rewriting:sec} integrates the lazy AC-matching in the
operational semantics of a rule application on a term, first at top position,
and then at other positions, through classical traversal strategies.
Section~\ref{implem:sec} presents a prototypal implementation of lazy
AC-matching and some experimental results derived from it.
Section~\ref{conclusion:sec} concludes.

\section{Notation and preliminaries}
\label{Notations:Sec}
Let $[n]$ denote the finite set of
positive integers $\{1,\ldots,n\}$ and let $|S|$ denote the cardinality of
a finite set $S$.  Thus, in particular, $|[n]| = n$.

Familiarity with the usual first-order notions of signature, (ground) term,
arity, position and substitution is assumed. Let $\mathcal{X}$ be a countable set
of variables, $\mathcal{F}$ a countable set of function symbols, and
$\mathcal{F}_{AC} \subseteq \mathcal{F}$ the set of associative-commutative
function symbols. Let $\EuScript{T}$
denote the set of terms built out of the symbols in $\mathcal{F}$ and the
variables in $\mathcal{X}$. Let $\EuScript{S}$ denote the set of substitutions
$\{x_1\mapsto t_1,\ldots,x_n\mapsto t_n\}$ with variables $x_1, \ldots, x_n$ in
$\mathcal{X}$ and terms $t_1, \ldots, t_n$ in $\EuScript{T}$. If $t$ is a term
and $\sigma$ is a
 substitution then $\sigma(t)$ denotes the term that results from the application
 of $\sigma$ to $t$. Given a position $p$, the subterm of $t$ at position
 $p$ is denoted by $t_{\mid p}$. We shall write $t_{\epsilon}$ for the
 symbol at the root of term $t$, i.e. $t=t_{\epsilon}(t_1,\ldots,t_n)$.

A term $t$ in $\EuScript{T}$ is \emph{flat} if, for any
position $p$ in $t$, $t_{\mid p} = +(t_1, \ldots, t_n)$ for some symbol $+$ in
$\mathcal{F}_{AC}$ implies that the root symbol ${(t_{i})}_{\epsilon}$ of each
direct subterm $t_{i}$ ($1 \leq i \leq n$) is not $+$. We denote  by
$\numberAC{t}$ the number of AC-symbols in the term $t$.

\paragraph{$\mathbb{T}$-Matching.}
For an equational theory $\mathbb{T}$ and any two terms $t$ and $t'$  in
$\EuScript{T}$ we say that $t$ matches $t'$ modulo $\mathbb{T}$ and write
$\matchth{t}{t'}{\mathbb{T}}$ iff there exists a substitution $\sigma$ in
$\EuScript{S}$ s.t. $\mathbb{T}\models (\sigma(t) =t')$. In this paper the theory
$\mathbb{T}$ is fixed. It is denoted $AC$ and axiomatizes the associativity and
commutativity of symbols in $\mathcal{F}_{AC}$,
i.e. it is the union of the sets of axioms $\{t_1+t_2=t_2+t_1,\;
(t_1+t_2)+t_3=t_1+(t_2+t_3)\}$
when $+$ ranges over $\mathcal{F}_{AC}$.

\paragraph{Rule-based semantics}
The following sections define the semantics of AC-matching, rule application and
strategy application by rewriting systems composed of labeled rewriting rules of
the form $\texttt{rule\_label:} \cdot \leadsto \cdot$, where the rewrite relation
$\leadsto$ should not be confused with the relation $\rrule{}{}$ of the rewriting
language. This semantics is said to be ``rule-based''.

\section{AC-matching and surjections}
\label{surj:sec}
Let $+ \in \mathcal{F}_{AC}$ be some associative and commutative function
symbol. This section first relates a restriction of the pattern-matching problem
$\matchth{+(t_1,\ldots,t_k)}{+(u_1,\ldots,u_n)}{AC}$, where  $1 \le k \le n$,
with the set $S_{n,k}$ of surjective functions (\emph{surjections}, for short)
from $[n]$ to $[k]$.
Then a notation is provided  to replace surjections with their rank to simplify
the subsequent exposition.

\begin{definition}[Application of a surjection on a term]
\label{surj:arrang:def}
Let $n \ge k \ge 1$ be two positive integers, $u=+(u_1,\ldots,u_n)$ be a term
 and $s \in S_{n,k}$ be a surjection from $[n]$ to
 $[k]$. The \emph{application of}
$s$ \emph{on} $u$ is defined by $s(u) =
+(\alpha_1,\ldots,\alpha_k)$ where $\alpha_i= u_j$ if
$s^{-1}(\{i\})=\{j\}$ and $\alpha_i~=~+(u_{j_1},\ldots,
u_{j_m})$ if $s^{-1}(\{i\})=\{j_1,\ldots,j_m\}$ and $j_1 < \ldots < j_m$  with 
$m\ge 2$.
\end{definition}

\begin{example}
\label{surj:appli:ex}
The application of the surjection $s=\{1\mapsto 2, 2 \mapsto 2, 3\mapsto 3, 4
\mapsto 1 \}$ on the term $+(u_1$, $u_2$, $u_3$, $u_4)$ is
 $s(+(u_1,u_2,u_3,u_4))=+(u_4,+(u_1,u_2),u_3)$.
\end{example}

The following proposition, whose proof is omitted, relates a subclass of
AC-matching problems with a set of surjections.
\begin{proposition}
\label{prop:matching:surj}
Let $t$ and $u$ be two flat terms with the same AC-symbol at the root and
containing no other AC-symbol than the one at the root. Let $k$ (resp. $n$) be
the arity of the root of $t$ (resp. $u$). Then the matching problem
$\matchth{t}{u}{AC}$ and the conjunction of matching problems $\bigwedge_{s\in
S_{n,k}} \matchth{t}{ s(u)}{\emptyset}$ admit the same set of solutions.
\end{proposition}

For $n\ge k \ge 1$, an integer $\textit{rk}(s)$ in $\{1, \ldots, |S_{n,k}|\}$
can be associated one-to-one to each surjection $s$ in $S_{n,k}$. It is called
the \emph{rank} of $s$.  It will be used in the subsequent sections to iterate
over the set $S_{n,k}$.
For each term $u$ of arity $n \geq 1$, each $1 \leq k \leq n$, and each integer
$i$ in $\{1,\ldots,|S_{n,k}|\}$, let $\textit{unrk}(i)$ denote the surjection
with rank $i$. The application of the integer $i$ to the term $u$, denoted by
$i(u)$, is defined by $i(u) = \textit{unrk}(i)(u)$, where application of a
surjection to a term is defined by Def.~\ref{surj:arrang:def}.

\section{AC-matching}
\label{ac:match:sec}
This section defines an eager and a lazy semantics for pattern-matching modulo
associativity and commutativity. By ``eager'' we mean a rewriting system
specifying the computation of all the solutions to a given matching problem,
without any mechanism to delay the production of other solutions when a first one
is produced. The AC-matching execution steps are made explicit by introduction of
a syntactic representation of matching problems by matching constraints. They are
a generalization to AC-matching of the matching constraints defined e.g.
in~\cite{CFK07} for syntactic matching (i.e. for the empty theory). Both
semantics of AC-matching are given in terms of a conditional rewrite system on
these constraints. In what follows all the terms are assumed to be flat.

\subsection{Eager AC-matching}
Figure~\ref{Eager:Matching-AC} proposes a rule-based \emph{eager} semantics for
AC-matching. This system (named \textit{Eager}) reduces \emph{eager matching
constraints} inductively defined by the grammar $\Eu{E} ::= \textbf{F} \;\;|\;\;
\textbf{I} \;\;|\;\;
 \matchAC{\Eu{T}}{\Eu{T}} \;\;|\;\; \Eu{E} \land \Eu{E} \;\;|\;\; \Eu{E} \lor
 \Eu{E}$.
This definition and the notations for matching constraints adapt and extend  the
ones from the $\rho$-calculus with explicit matching~\cite{CFK07}. The
constraint \textbf{F} denotes
 an absence of solution (failure). The constraint \textbf{I} denotes the
 identity substitution resulting from an initial trivial matching problem. The
 expression $\matchAC{p}{t}$ denotes the
elementary matching problem whether the term $t$ matches the pattern $p$ modulo
$AC$. The symbol $\land$ is the constructor of conjunctions of
 constraints. The symbol $\lor$ is introduced to enumerate various solutions as
 lists of constraints. Its priority is lower than that of $\land$. Both are
 assumed associative. Then the notation $C_1 \land \ldots \land C_n$ is used
 without ambiguity, and similarly for $\lor$, for $n \geq 1$. The constraint
 \textbf{I} is a neutral element for $\land$. The constraint \textbf{F} is a
 neutral element for $\lor$ and an absorbing element for $\land$.

 \begin{figure}[hbt!]
\setlength{\abovecaptionskip}{-0.2cm}
\setlength{\belowcaptionskip}{-0.1cm}
\begin{align*}
 \begin{array}{|l  l  l|}
\hline  
\texttt{E\_match\_AC:} 
  & \matchAC{\underbrace{+(t_1,\ldots,t_k)}_{t}}
    {\underbrace{+(u_1,\ldots,u_n)}_{u}}
  \leadsto 
  \bigvee_{j = 1}^{j=|S_{n,k}|} \bigwedge_{i = 1}^{i=k}
  \matchAC{t_i}{\alpha_i} & \\ 
  & \textrm{if } + \in \mathcal{F}_{AC}, \; k\le
  n \; \textrm{ and } j(u)=+(\alpha_1,\ldots,\alpha_k) & \\
&&\\
\texttt{E\_match:}
  & \matchAC{t_{\epsilon}(t_1,\ldots,t_n)}{t_{\epsilon}(u_1,\ldots,u_n)}
  \leadsto \bigwedge_{i = 1}^{i = n}\matchAC{t_i}{u_i}
  & \textrm{if }t_{\epsilon}\in \mathcal{F} \setminus \mathcal{F}_{AC} \\ 
&&\\
\texttt{E\_match\_AC\_fail:} 
  & \matchAC{+(t_1,\ldots,t_k)}{+(u_1,\ldots,u_n)} \leadsto \textbf{F} 
  & \textrm{if } + \in \mathcal{F}_{AC} \textrm{ and } k > n \\ 
&&\\
\texttt{E\_match\_fail:}
  & \matchAC{t_{\epsilon}(t_1,\ldots,t_n)}{u_{\epsilon}(u_1,\ldots,u_m)}
     \leadsto \textbf{F} 
  & \textrm{if } 
      t_{\epsilon}\neq u_{\epsilon} \\ 
&&\\
\texttt{E\_fail\_gen:} 
  & \matchAC{x}{t} \land E \land  \matchAC{x}{t'} \leadsto \textbf{F}  
   & \textrm{if } x\in \mathcal{X} \textrm{ and } t \neq t' \\
&&\\
\texttt{E\_var\_clash:} & \matchAC{t_{\varepsilon}(t_1,\ldots,t_n)}{x}  \leadsto \textbf{F}  & 
 \textrm{if }  x\in \mathcal{X}\\
&&\\
\texttt{E\_DNF\_1:} & F \land (G\lor H) \leadsto (F \land G) \lor (F \land H)&\\
\texttt{E\_DNF\_2:} &  (G\lor H) \land F \leadsto (G \land F) \lor (H \land
F)&\\
\hline
\end{array}
\end{align*}
\caption{\textit{Eager} system of AC-matching rules \label{Eager:Matching-AC}}
\end{figure}

In Figure~\ref{Eager:Matching-AC}, the notation $\bigwedge_{i=1}^{i=n}
\matchAC{t_i}{u_i}$ stands for \textbf{I} if $n = 0$ and for $\matchAC{t_1}{u_1}
\land \ldots \land \matchAC{t_n}{u_n}$ otherwise. The rule \texttt{E\_match\_AC}
corresponds to Proposition~\ref{prop:matching:surj} when $t_1$, \ldots, $t_k$,
$u_1$, \ldots, $u_n$ contain no AC symbol, and generalizes it otherwise.  The
positive integer $j$ iterates over surjection ranks. The rules \texttt{E\_match},
\texttt{E\_match\_fail}, \texttt{E\_fail\_gen}, and \texttt{E\_var\_clash} are
the same as in syntactic pattern-matching. With AC symbols, they are completed
with the rule \texttt{E\_match\_AC\_fail}. The rules \texttt{E\_DNF\_1} and
\texttt{E\_DNF\_2} correspond to the normalization of constraints into a
disjunctive normal form, DNF for short. A constraint is in DNF if it is of
the form $\lor_{i}\land_j F_{i,j}$, where $F_{i,j}$ is a constraint not
containing $\lor$ and $\land$.

It is standard to show that the system \textit{Eager} is terminating. The rules
\texttt{E\_DNF\_1} and \texttt{E\_DNF\_2} make it not confluent, but the
following post-processing reduces to a unique normal form all the irreducible
constraints it produces from a given pattern-matching problem. The
post-processing consists of \emph{(i)} replacing the trivial constraints of the
form $\matchAC{x}{x}$ by \textbf{I}, \emph{(ii)} replacing each non-trivial
constraint $\matchAC{x}{t}$ by the elementary substitution
 $x\mapsto t$, \emph{(iii)} eliminating duplicated elementary substitutions by
 replacing each expression of the form $\land_{i=1}^{i=n} E_i$ (with $n \ge 1$)
 by the set $\bigcup_{i=1}^{i=n} \{E_i\}$ that represents a
non-trivial substitution, then \emph{(iv)} replacing $\textbf{I}$ by $\{\,\}$
and finally \emph{(v)} replacing each disjunction  of
substitutions $\lor_{j=1}^{j=n} S_j$ with $n \ge 1$ by the set
$\bigcup_{j=1}^{j=n} \{S_j\}$, that represents a set of substitutions. A constraint in normal form can be either
\textbf{F}, if there is no solution to the initial matching problem, or a
non-empty set $\set{\sigma_1,\ldots,\sigma_n}$ of substitutions  which
corresponds to the set of all solutions of the matching problem. In particular
$\sigma_i$ may be $\{\,\}$, for some $i\in [n]$.

The theory for the associative symbols $\land$ and $\lor$ deliberately excludes
commutativity, because they appear in the pattern of some rules of the
\textit{Eager} system. We now motivate  this design choice. 
Since the \textit{Eager} system  is terminating and confluent 
-- in the sense explained above -- we can consider
it  as an algorithm and implement it without
modification, in a programming language supporting pattern-matching modulo the
theory of constraint constructors. Thus it would be ill-founded to require that
the language supports AC-matching, in order to extend it precisely with an
AC-matching algorithm! In this aspect our approach differs from the one
of~\cite{rhoCalIGLP-I+II-2001}, whose more compact calculus handles ``result
sets'', i.e. the underlying theory includes associativity, commutativity and
idempotency, but matching with patterns
containing set constructors is implemented by explorations of set data
structures.

\subsection{Lazy AC-matching}
\label{ac:matching:sec}

We now define a lazy semantics for pattern-matching modulo associativity and
commutativity, as a rewriting system named \emph{Lazy}. It reduces constraints
defined as follows.

\begin{definition}
The set of \emph{delayed matching
constraints}, hereafter called \emph{constraints} for short, is inductively
defined by the grammar: $ \Eu{C} ::= \textbf{F} \;\;|\;\; \textbf{I}
\;\;|\;\;
 \matchAC{\Eu{T}}{\Eu{T}} \;\;|\;\; \Eu{C} {\land} \Eu{C} \;\;|\;\;
  \Eu{C} \lor \Eu{C}
    \;\;|\;\;   \textit{Next}(\Eu{C}) \;\;|\;\;    
    \trip{\Eu{T},\Eu{T},\mathbb{N}^{\star}}.$
\end{definition} 
The first five constructions have the same meaning as in eager matching
constraints. As in the eager case, the symbols $\land$ and $\lor$ are
associative, the constraint \textbf{F} is an absorbing element for $\land$, and
the constraint \textbf{I} is a neutral element for  $\land$. However,  the
constraint \textbf{F} is no longer a neutral element for $\lor$.
 The construction $\textit{Next}(C)$ serves to activate the delayed computations
 present in the constraint $C$.
 When the terms $t$ and $u$ have the same AC symbol at the root, the constraint
 $\trip{t,u,s}$ denotes
the delayed matching computations of the problem $\matchAC{t}{u}$ starting with
the surjection 
with rank $s$, and hence, the matching computations for all the
surjections 
with a rank $s'$ s.t. $s' < s$ have already been performed. The conditions
that $t$ and $u$ have the same AC symbol at the root and that the arities of $t$
and $u$ correspond to the domain and codomain of 
the surjection with rank $s$ are not made
explicit in the grammar, but it would be easy to check that they always hold by
inspecting the rules of the forthcoming system \emph{Lazy}.
Delayed matching constraints of the form $\trip{t,u,s}$ and satisfying these
conditions are more simply called \emph{triplets}.

Figure~\ref{Matching-AC} defines the first part of the lazy
semantics, a rewriting system named $\mathcal{R}_1$.
The rules \texttt{match\_AC\_fail}, \texttt{match},  \texttt{match\_fail}, 
 \texttt{var\_clash} and \texttt{fail\_gen} are standard and
already appeared  in the eager matching system.
The rule \texttt{match\_AC} activates the delayed matching computations starting from 
the first surjection. It is immediately followed by the rule
\texttt{match\_surj\_next} from the rewriting system $\mathcal{R}_2$ defined in
Figure~\ref{Next}.

\begin{figure}[htb!]
\begin{align*}
 \begin{array}{@{}|l  l  l|@{}}
\hline  
\texttt{match\_AC:} 
&\matchAC{\underbrace{+(t_1,\ldots,t_k)}_{t}}{\underbrace{+(u_1,\ldots,u_n)}_{u}}
  \leadsto {\textit{Next}(\trip{t,u,1})}  
 & \textrm{ if } + \in \mathcal{F}_{AC} \textrm{ and } \;  k\le n \\
\texttt{match:} 
 &\matchAC{t_{\epsilon}(t_1,\ldots,t_n)}{t_{\epsilon}(u_1,\ldots,u_n)} \leadsto 
  \bigwedge_{i = 1}^{i = n} \matchAC{t_i}{u_i}
  & \textrm{ if }t_{\epsilon}\in \mathcal{F} \setminus \mathcal{F}_{AC}\\
&&\\
\texttt{match\_AC\_fail:} 
 &\matchAC{+(t_1,\ldots,t_k)}{+(u_1,\ldots,u_n)} \leadsto 
  \textbf{F} 
&\textrm{ if } + \in \mathcal{F}_{AC} \textrm{ and } k > n \\
&&\\
\texttt{match\_fail:}
 & \matchAC{t_{\epsilon}(t_1,\ldots,t_n)}{u_{\epsilon}(u_1,\ldots,u_m)}
    \leadsto \textbf{F} 
   & \textrm{ if } t_{\epsilon}\neq u_{\epsilon}\\ 
&&\\
\texttt{fail\_gen:} 
 &\matchAC{x}{t}\land C \land \matchAC{x}{t'} \leadsto \textbf{F} 
   & \textrm{ if }  x\in \mathcal{X} \textrm{ and } t \neq t'\\
&&\\
\texttt{var\_clash:} 
  & \matchAC{t_{\varepsilon}(t_1,\ldots,t_n)}{x}  \leadsto \textbf{F}  
   & \textrm{ if }  x\in \mathcal{X} \\
\hline
\end{array}
\end{align*}
\caption{$\mathcal{R}_1$ system: AC-matching rules \label{Matching-AC}}
\end{figure}

\begin{figure}[hbt!]
\begin{align*}
\begin{array}{@{}|l@{\;}l@{}l@{\;}|@{}}
\hline
 \texttt{fail\_next:} &\textbf{F} \lor C \leadsto \textit{Next}({C}) \hspace{2.9cm}
\texttt{next\_fail:} \hspace{1cm} \textit{Next}(\textbf{F}) \leadsto \textbf{F} &\\
\texttt{next\_id:}&  \textit{Next}(\textbf{I}) \leadsto \textbf{I} \hspace{3.7cm}
 \texttt{next\_basic:} \hspace{0.75cm}  \textit{Next}(\matchAC{x}{u}) \leadsto \matchAC{x}{u} & \\
&& \\
\texttt{next\_and:}&  \textit{Next}(C_1 \land C_2) \leadsto 
  \textit{Next}(C_1) \land \textit{Next}(C_2)  &\\
\texttt{next\_or:}& \textit{Next}(C_1 \lor C_2) \leadsto  \textit{Next}(C_1)\lor C_2, 
   \;\; \textrm{ if } C_1 \neq \textbf{F}  &\\
&& \\
\texttt{match\_surj\_next:}& \textit{Next}(\trip{t,u,{s}}) \leadsto
(\bigwedge_{i=1}^{i=k}\matchAC{t_i}{\alpha_i}) \lor {\trip{t,u,{s}+1}} &  \\
 & \textrm{ if }  t=t_{\epsilon}(t_1,\ldots,t_k), \;
 u=t_{\epsilon}(u_1,\ldots,u_n), \; {s} < |S_{n,k}| \textrm{ and } 
 {s}(u)=t_{\epsilon}(\alpha_1,\ldots,\alpha_k)    &\\ 
&&\\
\texttt{match\_surj\_last:}& \textit{Next}(\trip{t,u,|S_{n,k}|}) 
\leadsto \bigwedge_{i=1}^{i=k}\matchAC{t_i}{\alpha_i} &\\
 & \textrm{ if } 
     t=t_{\epsilon}(t_1,\ldots,t_k), \; 
     u=t_{\epsilon}(u_1,\ldots,u_n)  \textrm{ and }
     |S_{n,k}|(u)=t_{\epsilon}(\alpha_1,\ldots,\alpha_k)&\\
\hline 
\end{array}
\end{align*}
\caption{$\mathcal{R}_2$ system: \textit{Next} rules \label{Next}}
\end{figure}
In $\mathcal{R}_2$ the rule \texttt{fail\_next} states that the presence of a
failure activates the delayed computations in the constraint $C$. The other rules
propagate the activation of the delayed computations on the inductive structure
of constraints. The rule \texttt{next\_and}
 propagates the \textit{Next} constructor to sub-constraints. The rule
 \texttt{next\_or} propagates the $\textit{Next}$ constructor to the
 head of a list of constraints, provided this head is not \textbf{F}.

When the constraint is $\trip{t,u,s}$, two cases have to be considered.  If the
surjection rank $s <|S_{n,k}|$ is not the maximal one, then the rule
\texttt{match\_surj\_next} reduces the constraint
$\textit{Next}(\trip{t,u,{s}})$ to a set of matching constraints according to
the surjection with rank $s$,
 followed by  delayed computations that will
be activated from the next surjection, with rank $s+1$.
In the final case when $s=|S_{n,k}|$ (rule \texttt{match\_surj\_last}), there is
no delayed computations.
 
\begin{figure}[hbt!]
\begin{align*}
\begin{array}{|l l l|}
\hline
&\texttt{DNF\_1:}&  F \land (G\lor H) \leadsto (F \land G) \lor (F \land H)\\
&\texttt{DNF\_2:} &  (G\lor H) \land F \leadsto (G \land F) \lor (H \land F)\\
\hline
\end{array}
\end{align*}
\caption{$\mathcal{R}_3$ system : DNF  rules \label{Simplify}}
\end{figure}

Rules for the reduction of constraints in DNF are the same as in the eager
system, but are defined in the separate rewrite system $\mathcal{R}_3$ given in
Figure~\ref{Simplify}.

Let us denote by $\textit{Lazy}$ the
rewriting system $\mathcal{R}_1\cup \mathcal{R}_2 \cup \mathcal{R}_3$ composed of the 
 rules of Figures~\ref{Matching-AC}, \ref{Next} and
 \ref{Simplify}, equipped with the evaluation strategy \emph{sat} that consists of 
the iteration of the following process:
\emph{(i)} Applying the rules of $\mathcal{R}_1\cup \mathcal{R}_2$ until no
rule is applicable, then \emph{(ii)} applying the rules of $\mathcal{R}_3$
until no rule is applicable.

\subsection{Termination of lazy AC-matching}
We prove the termination of the lazy AC-matching system $\textit{Lazy}$
modularly, by considering some of its sub-systems separately, in the following lemmas. 

\begin{lemma}\label{R2:termination:lemma}
The rewriting system $\mathcal{R}_2$ is terminating.
\end{lemma}
\begin{proof}
On the one hand, no rule in $\mathcal{R}_2$ produces the \textbf{F}
constraint. Thus, after replacing with $\textit{Next}(C)$ the occurrences of
$\textbf{F} \lor C$ in the input constraint, the rule \texttt{fail\_next} is no longer used.
On the other hand, no rule in $\mathcal{R}_2$ can reduce the right side of the
rules \texttt{match\_surj\_next} and \texttt{match\_surj\_last}. 
Thus it is sufficient to prove the termination of 
${\mathcal{R}}_4 = \{\texttt{next\_fail}$, $\texttt{next\_id}$,
 $\texttt{next\_basic}$, $\texttt{next\_and}$, $\texttt{next\_or}\}$.
This is standard, using a recursive path ordering~\cite{Dersh_Ordering82}.
Intuitively, the termination of ${\mathcal{R}}_4$ is ensured by the fact that all the rules
in ${\mathcal{R}}_4$ push the $\textit{Next}$ constructor down until reaching the leaves. 
\end{proof}

 Before proving the termination of
$\mathcal{R}_1\cup \mathcal{R}_2$, we need to introduce a variant of
terms and triplets to prove a technical lemma about their occurrences in
derivations.
The \emph{marked term} $t^{M}$ associated to the term $t$ is obtained from $t$ by
replacing each AC symbol $+$ at some position $p$ in $t$ with $+_p$. 
The
\emph{marked triplet} associated to the triplet $\trip{t,u,s}$ is defined to be
$\trip{t^M,u,s}$. The \emph{marked constraint} associated to the
constraint $\matchAC{t}{u}$ is $\matchAC{t^M}{u}$.
The \emph{marked variant} $\mathcal{R}_1^{M}$ (resp.
$\mathcal{R}_2^{M}$, $\textit{Lazy}^{M}$) of the $\mathcal{R}_1$ (resp.
$\mathcal{R}_2$, \textit{Lazy}) system is obtained from the latter by
replacing triplets with marked triplets, constraints with marked constraints and
$+(t_1,\ldots,t_k)$ by $+_p(t_1,\ldots,t_k)$ for any $p$, in the pattern  of the rules
\texttt{match\_AC} and \texttt{match\_AC\_fail}. It is clear, and thus admitted,
that the derivations of $\textit{Lazy}^{M}$ are the variants of the derivations
of $\textit{Lazy}$, in a natural sense.

\begin{lemma}\label{bounded:triplet:lemma}
Let $t^M$ be a marked term with an AC symbol at the root and $u$ be a term with
the same AC symbol at the root. Consider a derivation $ \matchAC{t^M}{u}\leadsto
C_1 \leadsto C_2  \leadsto
 \ldots$ with rules in $\mathcal{R}_1^{M}\cup \mathcal{R}_2^{M}$. Then, the
 number of marked triplets in the sequence $C_1,C_2,\ldots$ is strongly bounded,
 in the following sense:
\emph{(i)} There is an upper bound for the number of marked triplets in each
$C_i$ and \emph{(ii)} for each subterm $t'$ of $t$ with an AC symbol at the root,
if a marked triplet built up on the marked subterm $t'$ (i.e. of the form
$\trip{t',u',s}$ for some $u'$ and $s$) is deleted from some $C_i,i\ge 1,$ then
it never appears again in $C_j$, for all $j>i$.
\end{lemma}

\begin{proof}
(i) The variant of the rule \texttt{match\_AC} replaces an elementary matching
problem (between two terms with an AC symbol at the root) with a marked triplet.
It is the only rule of $\mathcal{R}_1^{M} \cup \mathcal{R}_2^{M}$ that produces a
marked triplet whose first (marked) term is new. Then the variant of the rule
\texttt{match\_surj\_next} replaces a marked triplet whose first (marked) term is
the subterm $t' = +_p(\ldots)$ of $t$ at some position $p$ with a marked
triplet on the same terms, whose surjection rank is incremented. Thus
the number of marked triplets in a constraint is bounded above by the number of
positions of AC symbols in $t$.

(ii) The only rule of $\mathcal{R}_1^{M} \cup \mathcal{R}_2^{M}$ that deletes a
marked triplet whose first (marked) term is $t'$ is the variant of
\texttt{match\_surj\_last}. This marked triplet never appears again because the
variant of the rule \texttt{match\_surj\_next} just increments the surjection
rank of remaining triplets, whose first marked term $t''$ is another subterm of
$t$, located at another position in $t$ ($t''$ at least differs from $t'$ by the
name of its root symbol), and the other rules of $\mathcal{R}_1^{M} \cup
\mathcal{R}_2^{M}$ neither create new marked triplets nor duplicate existing
ones.
\end{proof}

\begin{lemma}\label{R12:termination:lemma}
The rewriting system ${\mathcal{R}_1} \cup \mathcal{R}_2$ is terminating.
\end{lemma}

\begin{proof}
Notice that $\mathcal{R}_1 \setminus \set{\texttt{match\_AC}} $ is clearly
terminating since it is a subsytem of the eager one, which is known
to be terminating.  We deduce that ${\mathcal{R}}_1$ is terminating since no
rule in ${\mathcal{R}}_1$ reduces the right side of the rule 
\texttt{match\_AC}.
Since $\mathcal{R}_1$ and $\mathcal{R}_2$ are terminating, it remains to show
that there is no infinite  reduction that goes back and forth between
$\mathcal{R}_1$ and $\mathcal{R}_2$. Toward a contradiction, assume that there is
an infinite reduction $C_1 \leadsto^+_{\mathcal{R}_1} C_2
\leadsto^+_{\mathcal{R}_2} C_3 \leadsto^+_{\mathcal{R}_1}\ldots$ that goes back
and forth between ${\mathcal{R}_1}$ and ${\mathcal{R}_2}$, where $C_1$
 is a pattern-matching problem. In $\mathcal{R}_2$ the rule \texttt{match\_surj\_next}
is the only rule producing new redexes for the system $\mathcal{R}_1$, i.e.
pattern-matching problems. Then the rule \texttt{match\_surj\_next} should
appear infinitely often in this infinite reduction. Equivalently, we consider
the marked variant of this derivation, with the same notations. Since the
number of marked triplets in the sequence $C_1,C_2,\ldots$ is strongly bounded (by
Lemma~\ref{bounded:triplet:lemma}) there is a marked triplet $(t,u,s_{k_i})$ and
an infinite sub-sequence $C_{k_1},C_{k_2},\ldots$ of $C_1,C_2,\ldots$ such that
$\trip{t,u,s_{k_i}} \in C_{k_i}$ and $s_{k_{i+1}}=s_{k_{i}}+1$ for each $i$.
This is a contradiction since the rank of surjections is upper bounded.
\end{proof}

\begin{theorem}
The $\textit{Lazy}$ system  is terminating.
\label{termination:theorem}
\end{theorem}

\begin{proof}
On the one hand, the termination of  $\mathcal{R}_1 \cup \mathcal{R}_2$ is proved
in Lemma \ref{R12:termination:lemma}. On the other hand, it is standard to show
that  the system  $\mathcal{R}_3$ is terminating. It corresponds to the
normalization w.r.t. to the disjunctive normal  form. Therefore, it is sufficient
to  show  that there is no infinite  reduction that goes back and forth between
the saturation of $\mathcal{R}_1\cup \mathcal{R}_2$ and the saturation of
$\mathcal{R}_3$. Let $\Eu{R}=\matchAC{t}{u}\leadsto^{\omega}_{\mathcal{R}_1\cup
\mathcal{R}_2} C_1 \leadsto^{\omega}_{\mathcal{R}_3} C_2
\leadsto^{\omega}_{\mathcal{R}_1\cup \mathcal{R}_2}\ldots$ be a reduction in
$\textit{Lazy}$. Notice that each $C_{p}$, where $p$ is even, is of the form $\bigvee_{i=1}^{i=k}F_i$
where $F_i=\bigwedge_{j=1}^{j=m}D_j$ and each $D_j$ is either a triplet or the
matching problem of a term with a variable. If there is no new redex
in $C_p$ then the reduction  $\Eu{R}$ stops. Otherwise, by observing the right side of the rules of $\mathcal{R}_1\cup
 \mathcal{R}_2$, we claim that if there is a new redex in $C_p$ -- which is
 created by the system $\mathcal{R}_3$ -- then this redex is  necessarily of the form
$\matchAC{x}{t}\land \ldots \land \matchAC{x}{t'}$ with $t\neq t'$, producing
 the \textbf{F} constraint. Let $q \in [k]$ be the smallest integer such that
 such a redex appears in $F_{q}$. If $q=k$ then $C_p$ is reduced to
 $\bigvee_{i=1}^{i=k-1} F_i \lor \textbf{F}$ and the reduction $\Eu{R}$
 terminates. Otherwise, 
$
C_{p} \leadsto_{\texttt{fail\_gen}} 
  F_1 \lor \ldots \lor
\textbf{F}\lor F_{q+1}\lor \ldots \lor \textbf{F}\lor \ldots \lor F_k 
\leadsto_{\texttt{fail\_next}}
F_1 \lor \ldots \lor \textit{Next}(F_{q+1}) \lor \ldots \lor \textbf{F}\lor \ldots \lor F_k.
$

The rules of $\mathcal{R}_2$  push the $\textit{Next}$ constructor down, and all
the  constraints of the form $\textit{Next}(\trip{t,u,s})$  in $F_{p+1}$ are
reduced to $(\bigwedge_{i=1}^{i=k}\matchAC{t_i}{\alpha_i}) \lor
\trip{t,u,{s}+1}$ by the rule \texttt{match\_surj\_next}, if
$s(t)=t_{\varepsilon}(\alpha_1,\ldots,\alpha_k)$. To prove that the reduction
$\Eu{R}$ terminates it is sufficient to prove that the system
$\mathcal{R}=\mathcal{R}_1\setminus\set{\texttt{match\_AC}}\cup
\set{\texttt{always\_next},\texttt{match\_AC\_2}} \cup \mathcal{R}_3$ is
terminating, with the following rule definitions:
$$
\texttt{match\_AC\_2: } \matchAC{t}{u} \leadsto \trip{t,u,1} \quad \textrm{ and 
} \quad \texttt{always\_next: }  \trip{t,u,{s}} \leadsto
(\bigwedge_{i=1}^{i=k}\matchAC{t_i}{\alpha_i}) \lor {\trip{t,u,{s}+1}}
.$$
However the termination of $\mathcal{R}$ is ensured by the termination of  
the eager system. That is, we have just replaced the rule \texttt{E\_match\_AC}
of \textit{Eager} with the rules \texttt{match\_AC\_2} and
\texttt{always\_next} in $\mathcal{R}$.
\end{proof}

\subsection{Confluence of lazy AC-matching}
\label{confluence:sec}
The system \textit{Lazy} is not confluent, due to the non-confluence of
$\mathcal{R}_3$. In this section we argue that the system $\mathcal{R}_1\cup
\mathcal{R}_2$ is confluent, and we consider an evaluation strategy
for $\mathcal{R}_3$ to  get a confluent AC-lazy matching system, that we call \lazy.

\begin{proposition}
The system $\mathcal{R}_1 \cup \mathcal{R}_2$ is locally confluent.
\end{proposition}
\begin{proof}
It is straightforward to check that there is no critical  overlap between 
 any two redexes, i.e. the contraction of one redex does not destroy the others. 
It is worth mentioning that without the  condition $C_1\neq \textbf{F}$ 
  of the rule \texttt{next\_or} 
we could have non-convergent critical pairs, e.g.
$Next(\textbf{F} \lor C)\leadsto Next(Next(C))$ by the rule 
\texttt{fail\_next} and  
$Next(\textbf{F} \lor C)$ $\leadsto$ $Next(\textbf{F}) \lor C$ 
$\leadsto$ $\textbf{F}\lor C$ $\leadsto$ $Next(C)$.
\end{proof}

\begin{corollary}
The system $\mathcal{R}_1\cup \mathcal{R}_2$
is confluent. 
\end{corollary}

 The reason of the non-confluence of 
$\mathcal{R}_3$  is the non-commutativity 
of the operators $\land$ and $\lor$.
It is classical to add a strategy to $\mathcal{R}_3$
so that the resulting system becomes confluent.
\begin{definition}
Let $\mathcal{R}_3^{\downarrow}$  be the system $\mathcal{R}_3$ 
 with the following strategy: 
\textit{(i)} When reducing a constraint of the 
form $\bigwedge_{i=1}^{i=k} \bigvee_j C_{i,j}$ with $k\ge 3$, first reduce
 $\bigwedge_{i=2}^{i=k} \bigvee_j C_{i,j}$, and 
\textit{(ii)}  reduce 
$(\bigvee_{l=1}^{l=m} A_l) \land B$  to 
$(A_1 \land B) \lor ((\bigvee_{l=2}^{l=m} A_l) \land B)$.
\end{definition}
\begin{proposition}(Admitted)
$\mathcal{R}_3^{\downarrow}$ is  confluent.
\end{proposition}
Now we are ready to define the lazy AC-matching.
\begin{definition}
\label{ac:conf:lazy:def}
The \emph{lazy AC-matching}, denoted by \lazy, is the
rewriting system $\mathcal{R}_1\cup \mathcal{R}_2 \cup \mathcal{R}_3$ composed of the 
 rules of Figures~\ref{Matching-AC}, \ref{Next} and
 \ref{Simplify}, equipped with the evaluation strategy $\emph{sat}^{\downarrow}$ 
that consists of the iteration of the following process:
\emph{(i)} Applying the rules of $\mathcal{R}_1\cup \mathcal{R}_2$ until no
rule is applicable, then \emph{(ii)} applying  $\mathcal{R}_3^{\downarrow}$ 
until no rule  is applicable.
\end{definition}

\begin{theorem}
\label{lazy:term:conf:th}
 \lazy is terminating  and confluent.
\end{theorem}
\begin{proof}
The termination of \lazy is a consequence of the one of \textit{Lazy}. The 
confluence of \lazy follows from the  confluence of $\mathcal{R}_1\cup
\mathcal{R}_2$ and $\mathcal{R}_3^{\downarrow}$.
\end{proof}

In what follows the normal form of a constraint $C$ w.r.t. 
a system $\mathcal{R}$  will be denoted by $\textit{NF}_{\mathcal{R}}(C)$, or just
$\textit{NF}(C)$ if $\mathcal{R}$ is \lazy.

\subsection{Normal forms and lazy lists}
\label{NF:subsec}
In this section we prove Theorem~\ref{NF:theorem} that characterizes  the normal forms
of the lazy AC-matching \lazy. They correspond basically to lazy lists. 
Roughly speaking, a lazy list is composed of a  substitution at the head
 and  a non-evaluated  object that represents the remaining substitutions.  
This characterization of the normal forms is of major importance since 
it guarantees   that  the  element at the head  is always a substitution. 
The formal definition of lazy lists follows.

\begin{definition}
\label{wfll:def}
A \emph{$\land$-substitution} is a conjunction of delayed matching constraints of
the form $\matchAC{x}{u}$ where $x$ is a variable. A $\land$-substitution is
\emph{irreducible} if it cannot be reduced by the rule \texttt{\emph{fail\_gen}}. A
constraint is called a \emph{lazy list} if it is \textbf{F}, \textbf{I}, an
irreducible $\land$-substitution or a constraint of the form $\sigma \lor C$
where $\sigma$ is an irreducible $\land$-substitution and  
$\textit{NF}(\textit{Next}(C))$ is also a lazy list.
\end{definition}

In order to characterize the normal forms of \lazy we first characterize in Lemma
\ref{R12:NF:lemma} the normal forms  of the system $\mathcal{R}_1\cup
\mathcal{R}_2$. Then we characterize in Lemma
 \ref{R123:NF:lemma} the normal forms of $\Rthree$
when it has  the normal forms of $\Ronetwo$ as input. Summing up these results,
we show in
 Proposition \ref{inviariance:FN:prop}
the invariance of the syntax of constraints after the composition of the
application of $\Ronetwo$ and of $\Rthree$. Finally, Theorem \ref{NF:theorem}
becomes an immediate consequence of Proposition \ref{inviariance:FN:prop}.

Let us begin by characterizing the normal forms of $\Ronetwo$ and introduce for
this purpose the grammar $\Eu{G}::=\Eu{G}\land \Eu{G}\;|\;
  \Eu{G} \lor \trip{\Eu{T},\Eu{T},\mathbb{N}^{*}}\;|\;
 \matchAC{\mathcal{X}}{\Eu{T}} \;|\; \textbf{I}$.

\begin{lemma}
\label{R12:NF:lemma}
The normal form of an AC-matching problem by the system $\mathcal{R}_1 \cup
\mathcal{R}_2$ is either \textbf{F} or follows the  grammar $\Eu{G}$.
\end{lemma}
\begin{proof}
Let  $\matchAC{t}{u}$ be an AC-matching problem. When
$\matchAC{t}{u}\leadsto_{\mathcal{R}_1 \cup \mathcal{R}_2} \textbf{F}$,
 the normal form of the AC-matching problem is $\textbf{F}$, since no
 rule rewrites $\textbf{F}$. The other cases when $t$ or $u$ is a variable
 are also trivial: The matching problem is reduced to \textbf{I} or is irreducible.
 
It remains to consider the case when $t=t_{\epsilon}(t_1,\ldots,t_k)$ and
$u=t_{\epsilon}(u_1,\ldots,u_n)$ for some $k,n \geq 0$.
The proof is by induction on the number of symbols in $t$.
If $t$ is a constant, i.e. $k = 0$, then $n=0$ and $u = t$. The AC-matching
problem $\matchAC{t}{u}$ is reduced to \textbf{I}, thus follows the  grammar
$\Eu{G}$.
Otherwise, $k \geq 1$. Only one rule can be applied, \texttt{match\_AC} or
\texttt{match}, depending on the nature of the symbol $t_{\epsilon}$ at the root of $t$.

\noindent \emph{Case 1.} If $t_{\epsilon}$ is an AC symbol, then $k \le n$ and
$\matchAC{t}{u} \leadsto_{\texttt{match\_AC}} \textit{Next}(\trip{t,u,1})$. We
prove more generally that the normal form of $\textit{Next}(\trip{t,u,s})$ by
$\mathcal{R}_1\cup \mathcal{R}_2$  is \textbf{F} or follows the grammar
$\Eu{G}$.  Let $s(u)=t_{\epsilon}(\alpha_1,\ldots,\alpha_k)$ and $C =
\bigwedge_{i=1}^{i=k}\matchAC{t_i}{\alpha_i}$.

If $|S_{n,k}|-s=0$, then $\textit{Next}(\trip{t,u,s})
\leadsto_{\texttt{match\_surj\_last}} C \leadsto^{\star} \textit{NF}(C)$. Otherwise, when
$|S_{n,k}|-s>0$, $\textit{Next}(\trip{t,u,s})
\leadsto_{\texttt{match\_surj\_next}} C \lor \trip{t,u,s+1} \leadsto^{\star}
\textit{NF}(C) \lor \trip{t,u,s+1}$.
In both cases, since each $t_i$ contains less symbols than $t$, the induction
hypothesis holds for each $t_i$, and hence the normal form $\textit{NF}(C)$ of $C =
\bigwedge_{i=1}^{i=k}\matchAC{t_i}{\alpha_i}$ is \textbf{F} or a constraint
which follows the grammar $\Eu{G}$, since \textbf{F} is an absorbing element
for $\land$.

For this case the remainder of the proof is by induction on $|S_{n,k}|-s$. The
basic case when $|S_{n,k}|-s=0$ has already been treated. When $|S_{n,k}|-s>0$,
there are two cases. If $\textit{NF}(C)$ follows the grammar $\Eu{G}$ then it obviously
also holds for $\textit{NF}(C) \lor \trip{t,u,s+1}$. Otherwise, $\textit{NF}(C)$ is \textbf{F}
and $\textit{Next}(\trip{t,u,s}) \leadsto^{\star} \textbf{F} \lor
\trip{t,u,s+1} \leadsto_{\texttt{fail\_next}} \textit{Next}(\trip{t,u,s+1})$.
Since $|S_{n,k}|-(s+1)<|S_{n,k}|- s$, the  induction hypothesis gives that
$\textit{NF}(\textit{Next}(\trip{t,u,s+1}))$ is \textbf{F} or follows the grammar
$\Eu{G}$, hence also for $\textit{Next}(\trip{t,u,s})$.

\noindent \emph{Case 2.}
If $t_{\epsilon}$ is not an AC symbol, then $k=n$ and $\matchAC{t}{u}$
$\leadsto_{\texttt{match}}$ $\bigwedge_{i = 1}^{i = k} \matchAC{t_i}{u_i}$
$\leadsto^{\star}$ $\bigwedge_{i = 1}^{i = k} \textit{NF}(\matchAC{t_i}{u_i})$.
By induction hypothesis, for each $i$, $\textit{NF}(\matchAC{t_i}{u_i})$ is 
\textbf{F} or follows the grammar $\Eu{G}$. If
$\textit{NF}(\matchAC{t_i}{u_i})=\textbf{F}$ for some $i \in[k] $, then
$\matchAC{t}{u}\leadsto^{*} \textbf{F}$, since \textbf{F} is an
absorbing element for $\land$. Otherwise, $\textit{NF}(\matchAC{t_i}{u_i})$ follows
$\Eu{G}$ for each $i$, and then it
obviously  also holds for their conjunction, and for $\textit{NF}(\matchAC{t}{u})$.
\end{proof}

\begin{lemma}
\label{stabl1:lemma}
Let $C$ be an irreducible constraint w.r.t. $\mathcal{R}_1\cup \mathcal{R}_2$
 that follows  the grammar $\Eu{G}$. Then, the normal form
of $\textit{Next}(C)$ by $\mathcal{R}_1\cup \mathcal{R}_2$ is $C$.
\end{lemma}

\begin{proof}
The proof is by induction on the grammar constructions of $\Eu{G}$.
If $C$ is \textbf{I} or a matching problem of the
form  $\matchAC{x}{u}$ where $x$ is a variable, then the rules
\texttt{next\_id} and \texttt{next\_basic} ensure that $\textit{Next}(C)\leadsto
C$.  If $C=C_1\land C_2$ then $\textit{Next}(C_1\land C_2)\leadsto
\textit{Next}(C_1)\land \textit{Next}(C_2)$ and the induction hypothesis 
  $\textit{NF}(\textit{Next}(C_1))=C_1$
and $\textit{NF}(\textit{Next}(C_2))=C_2$ apply to show that
$\textit{NF}(\textit{Next}(C_1\land C_2)) = \textit{NF}(C_1\land C_2)=C_1\land C_2$. 
If $C=C_1\lor \trip{t,u,s}$,
then $\textit{Next}(C_1\lor \trip{t,u,s}) \leadsto \textit{Next}(C_1)\lor \trip{t,u,s}
 \leadsto C_1 \lor \trip{t,u,s}$.
\end{proof}

We can generalize the previous lemma by the following one.
\begin{lemma}
\label{stabl2:lemma}
Let $F$ be a constraint of the form $\bigwedge_{i} F_i$, where each 
$F_i$ is either a triplet or follows 
the grammar $\Eu{G}$, such that $F$ is irreducible w.r.t.  
$\mathcal{R}_1\cup \mathcal{R}_2$.
Then, the normal form of $\textit{Next}(F)$ by $\mathcal{R}_1\cup \mathcal{R}_2$
is \textbf{F} or follows the  grammar $\Eu{G}$.
\end{lemma}
\begin{proof}
By iterating application of the rule \texttt{next\_and} we get
$\textit{Next}(\bigwedge_{i} F_i)\leadsto^{*}\bigwedge_{i}\textit{Next}(F_i)$.
$F_i$ is irreducible w.r.t. $\mathcal{R}_1\cup \mathcal{R}_2$  because $F$ is
irreducible. If $F_i$ follows the grammar $\Eu{G}$, then
$\textit{NF}(\textit{Next}(F_i))=F_i$ by Lemma \ref{stabl1:lemma}. If $F_i$ is a triplet  
then  by  Lemma \ref{R12:NF:lemma} the normal form of $\textit{Next}(F_i)$
is \textbf{F} or follows $\Eu{G}$. Therefore, the normal form of
 $\bigwedge_{i}\textit{Next}(F_i)$ is \textbf{F} or follows  the  grammar
 $\Eu{G}$.
\end{proof}

We define the grammar $\Eu{K}::=\Eu{K}
\land \Eu{K}\;|\; \matchAC{\Eu{X}}{\Eu{T}}\;|\;
\trip{\Eu{T},\Eu{T},\mathbb{N}^{\star}}$ for conjunctions of atomic constraints,
the grammar $\Eu{F}::=\Eu{F}\lor \Eu{F} \;|\; \Eu{K}$ for
constraints in DNF, the grammar $\Eu{S} ::= \Eu{S} \land \Eu{S} \;|\;
\matchAC{\Eu{X}}{\Eu{T}}$ for $\land$-substitutions,
and the grammar $\Eu{H}::=\Eu{S} \lor \Eu{F}\;|\; \Eu{S}\;|\; \textbf{I}\lor
\Eu{F}\;|\; \textbf{I}$ to formulate the following two lemmas. The first one is
about normal forms by $\Rthree$ of inputs which are normal forms of $\Ronetwo$.

\begin{lemma}
\label{R123:NF:lemma}
Let $C$ be an irreducible  constraint w.r.t. $\mathcal{R}_1\cup \mathcal{R}_2$
 that follows  the grammar $\Eu{G}$. Then, the normal form of $C$ by the system
 $\mathcal{R}_3^{\downarrow}$ follows the grammar $\Eu{H}$.
\end{lemma}
\begin{proof}
On the one hand, since $C$ follows the grammar $\Eu{G}$, it is built
up on $\land$, $\lor$, \textbf{I}, matching constraints of the form
$\matchAC{x}{u}$ and triplets. On the other hand, the normal form of $C$ by
$\mathcal{R}_3^{\downarrow}$  is in DNF. Therefore it is sufficient to show that
$\textit{NF}_{\mathcal{R}_3^{\downarrow}}(C)$ is either \textbf{I} or of the
form $\sigma \lor F$, where $\sigma$ is either a $\land$-substitution or
\textbf{I}, and $F$ follows  $\Eu{F}$. The  proof is by   induction on the
grammar constructions of $\Eu{G}$. If  $C$ is $\textbf{I}$ or $\matchAC{x}{u}$
then the claim holds.
Otherwise, we distinguish two cases:

\noindent \emph{Case 1.} If   $C=C_1 \land C_2$, then we only discuss the
non-trivial case when $\textit{NF}(C_1)$ or $\textit{NF}(C_2)$ is of the form $\sigma\lor F$.
Assume that  $C_1=\sigma_1\lor F_1$, the  other case
can be handled similarly.
In this case $C_2$ can be \textbf{I}, $\sigma_2$ or $\sigma_2\lor F_2$.
If $C_2=\textbf{I}$ then 
$(\sigma_1 \lor F_1)\land \textbf{I}=\sigma \lor F_1$. 
If $C_2=\sigma_2$, then 
$(\sigma_1 \lor F_1)\land \sigma_2 \leadsto 
(\sigma_1\land\sigma_2) \lor (F_1\land \sigma_2)$. 
Finally, if $C_2=\sigma_2 \lor
F_2$, then $(\sigma_1 \lor F_1)\land (\sigma_2 \lor F_2) \leadsto^{*}
      (\sigma_1 \land \sigma_2)\lor (\sigma_1 \land F_2) \lor 
      (F_1 \land\sigma_2) \lor (F_1 \land F_2)$ and the claim holds.

\noindent \emph{Case 2.} If $C=C_1 \lor \trip{t,u,s}$, then  we apply the
induction hypothesis on $\textit{NF}(C_1)$, and the desired result follows. 
\end{proof}

The following lemma describes the syntax of the result of
$\textit{Next}(F)$ by the application of $\Ronetwo$
followed by the application of $\Rthree$, when $F$ is an irreducible
constraint in DNF.

\begin{lemma}
\label{invariance:next:lemma}
Let  $\phi(p)=\bigvee_{i=1}^{i=p}\bigwedge_{j=1}^{j=q}C_{i,j}$ 
be an irreducible constraint
w.r.t. $\mathcal{R}_1\cup \mathcal{R}_2$ that follows the grammar $\Eu{F}$.
Then  $\textit{NF}_{ \mathcal{R}_3^{\downarrow}}(\textit{NF}_{\mathcal{R}_1\cup \mathcal{R}_2}(Next(\phi(p))))$ 
is either \textbf{F} or follows the grammar $\Eu{H}$.
\end{lemma}
\begin{proof}
The proof is by induction on $p$. If $p=1$ then by Lemma \ref{stabl2:lemma}
$\phi' = \textit{NF}_{\mathcal{R}_1\cup \mathcal{R}_2}(Next(\phi(1)))$ is \textbf{F} or
follows  $\Eu{G}$. Therefore $\textit{NF}_{\mathcal{R}_3^{\downarrow}}(\phi')$ is 
\textbf{F}, or  follows  $\Eu{H}$ by Lemma \ref{R123:NF:lemma}. If $p>1$
then $Next(\bigvee_{i=1}^{i=p}\bigwedge_{j=1}^{j=q}C_{i,j}) \leadsto
Next(\bigwedge_{j=1}^{j=q}C_{1,j}) \lor \phi(p-1)$.
 By  Lemma \ref{stabl2:lemma} again 
$\textit{NF}_{\mathcal{R}_1\cup \mathcal{R}_2}(Next(\bigwedge_{j=1}^{j=q}C_{1,j}))$
  is \textbf{F} or follows $\Eu{G}$. 
  
  In the first case, we apply
  $\textbf{F} \lor \phi(p-1) \leadsto Next(\phi(p-1))$, and use the
  induction hypothesis that $\textit{NF}_{\mathcal{R}_3^{\downarrow}}(\textit{NF}_{\mathcal{R}_1\cup
\mathcal{R}_2}(Next(\phi(p-1))))$ follows $\Eu{H}$.  In the second case, it
comes from Lemma \ref{R123:NF:lemma} that
 $\psi=\textit{NF}_{\mathcal{R}_3^{\downarrow}}(\textit{NF}_{\mathcal{R}_1\cup
 \mathcal{R}_2}(Next(\bigwedge_{j=1}^{j=q}C_{1,j})))$
follows $\Eu{H}$. Hence $\psi \lor (\phi(p-1))$ follows $\Eu{H}$, since
$\phi(p-1)$ follows $\Eu{F}$.
\end{proof}

Now we are ready to prove the following invariance proposition. It generalizes
the previous lemma by considering an arbitrary constraint following the grammar
$\Eu{H}$.

\begin{proposition}(Invariance proposition)
\label{inviariance:FN:prop}
Let $C$ be a constraint following the grammar $\Eu{H}$.
  Then 
  		\linebreak
  $\textit{NF}_{\Rthree}(\textit{NF}_{\Ronetwo}(C))$ is \textbf{F} or follows the grammar 
$\Eu{H}$.
\end{proposition}
\begin{proof}
The case when $C$ is a $\land$-substitution or \textbf{I} is trivial.
Otherwise, let $C=\sigma \lor F$, where $\sigma$ is \textbf{I} or a
$\land$-substitution and $F$ follows  $\Eu{F}$.
Notice that the only potential redexes in $C$  w.r.t. 
the system $\mathcal{R}_1\cup \mathcal{R}_2$ are of the form
$\matchAC{x}{u_1}\land \ldots \land \matchAC{x}{u_2}$ such that 
$u_1\neq u_2$. In this case the rule \texttt{fail\_gen} is
applied. Let us call such redexes \emph{failure redexes.}
We  distinguish two cases. 

\noindent \textit{Case 1.} If there is no 
 failure redex in $F$ (i.e. $F$ is irreducible w.r.t. $\Ronetwo$) 
then we again distinguish two cases. 
If there is no failure redex in $\sigma$, then we are done.
Otherwise, 
$\sigma \lor F \leadsto \textbf{F}\lor F \leadsto \textit{Next}(F)$,
and the result follows from Lemma \ref{invariance:next:lemma}.

\noindent \textit{Case 2.} If there are some failure redexes in $F$ then assume
that $F=\bigvee_{i=1}^{i=m}F_i$, and let $I=[m]$. Let us argue  that
$\textit{NF}_{\Ronetwo}(F)$  is either \textbf{F} or of the form $\bigvee_{i \in
I'}F'_i$ where either $F'_i=F_i$ or $F'_i=\textit{NF}_{\Ronetwo}(Next(F_i))\neq
\textbf{F}$ for some $I'\subseteq I$. We propose an algorithm to construct
$I'$. Let $I':=I$ initially. (a) If $m=1$, the expected form for
$\textit{NF}_{\Ronetwo}(F)$ is obtained with the current $I'$. Otherwise, if $m>1$, let
$p$ be the smallest integer in $[m]$ such that $F_i$ contains a failure redex.
Therefore we have $\bigvee_{i=1}^{i=m}F_i\leadsto \bigvee_{i=1}^{i=p-1}F_i \lor
\texttt{ fail} \lor  F_{p+1}\lor \bigvee_{i=p+2}^{i=m}F_i \leadsto
\bigvee_{i=1}^{i=p-1}F_i \lor  Next(F_{p+1})\lor \bigvee_{i=p+2}^{i=m}F_i$.
Continue the elimination of the failure redexes in $\bigvee_{i=p+2}^{i=m}F_i$ by
iterating (a) with $\{p+2, \ldots,m\}$ instead of $[m]$. Let $G$ by the
resulting disjunction. If $\textit{NF}_{\Ronetwo}(Next(F_{p+1}))\neq \textbf{F}$, then
let $I':=I'\setminus\set{p}$. Otherwise, if
$\textit{NF}_{\Ronetwo}(Next(F_{p+1}))=\textbf{F}$, then let
$I':=I'\setminus\set{p,p+1}$ and continue the reduction on
 $Next(G)$. 

Since  $F'_i=F_i$ or $F'_i=\textit{NF}_{\Ronetwo}(Next(F_i))$, where $i\in I'$,
then  
\begin{enumerate}
  \item[(i)] 
 $\textit{NF}_{\Rthree}(F'_i)$ either follows $\Eu{F}$ or
 follows $\Eu{H}$ by   Lemma \ref{R123:NF:lemma}. 
Therefore, $\textit{NF}_{\Rthree}(\bigvee_{i \in I'}F'_i)$ is either 
\textbf{F} or follows $\Eu{F}$ and $\textit{NF}_{\Rthree}(\textit{NF}_{\Ronetwo}(F))$ is either \textbf{F} or follows the grammar 
$\Eu{F}$. 
\item[(ii)] In order to simplify the computations,  
let  $q=|I'|$ and consider the renaming 
$\bigvee_{i\in [q]} H_i=\bigvee_{i\in I'} F'_i$. 
We argue by induction on $q$ that 
$\textit{NF}_{\Rthree}(\textit{NF}_{\Ronetwo}(Next(\bigvee_{i\in [q]}H_i)))$ 
is either \textbf{F} or follows $\Eu{H}$. 
If $q=1$, then by Lemma \ref{stabl2:lemma} it follows that 
$\textit{NF}_{\Ronetwo}(Next(H_i))$ is either \textbf{F} or 
 follows $\Eu{G}$ and therefore, by Lemma \ref{R123:NF:lemma},
we have that $\textit{NF}_{\Rthree}(\textit{NF}_{\Ronetwo}(Next(H_i)))$ is either 
\textbf{F} or follows $\Eu{H}$.
If $q>1$, then 
$Next(\bigvee_{i\in [q]}H_i)\leadsto Next(H_1) \lor \bigvee_{i=2}^{i=q}H_i$.
If $Next(H_1) \leadsto^{\star} \textbf{F}$, then we apply the induction 
hypothesis to 
$\textit{NF}_{\Rthree}(\textit{NF}_{\Ronetwo}(Next(\bigvee_{i=2}^{i=q}H_i)))$. Otherwise,
$\textit{NF}_{\Ronetwo}(Next(H_1))$ follows $\Eu{G}$, and hence 
$\textit{NF}_{\Rthree}(\textit{NF}_{\Ronetwo}(Next(H_1)))$ follows $\Eu{H}$. On the other 
hand, $\textit{NF}_{\Rthree}(\bigvee_{i=2}^{q}H_i)$ follows $\Eu{F}$,
by (i). Summing up, we get that $\textit{NF}_{\Rthree}(\textit{NF}_{\Ronetwo}(Next(H_1))) \lor 
\textit{NF}_{\Rthree}(\bigvee_{i=2}^{q}H_i))$ follows $\Eu{H}$.
\end{enumerate}

Now we distinguish two cases for $\sigma$.  If there is no failure redex in
$\sigma$, then by (i) we get that $\sigma \lor \textit{NF}_{\Rthree}(\textit{NF}_{\Ronetwo}(F))$
is either a $\land$-substitution or follows $\Eu{H}$. 
Otherwise, if there are some failure redexes in $\sigma$, then
 we get $\sigma \lor F \leadsto^{\star} \textbf{F}\lor \textit{NF}_{\Ronetwo}(F) 
\leadsto Next(\textit{NF}_{\Ronetwo}(F))$. From (ii)
it follows that $\textit{NF}_{\Rthree}(Next(\textit{NF}_{\Ronetwo}(F)))$ is  \textbf{F}
or follows $\Eu{H}$.
 \end{proof}

\begin{theorem}
\label{NF:theorem}
The normal form of a pattern-matching constraint $C$ by the system 
 \lazy  is a lazy list.
\end{theorem}
\begin{proof}
From  the termination of \lazy (Theorem~\ref{lazy:term:conf:th}) and
Proposition~\ref{inviariance:FN:prop}, we deduce that the normal form of \lazy is
either \textbf{F} or  follows the grammar $\Eu{H}$ and does not contain any
failure redex. Such a normal form is of the form $\sigma$ or $\sigma \lor F$,
where $\sigma$ is either \textbf{I} or an irreducible $\land$-substitution.
Therefore, it remains to show that $\textit{NF}(Next(F))$ is a lazy list,
or, equivalently, that $\textit{NF}_{\Rthree}(\textit{NF}_{\Ronetwo}(Next(F)))$ follows $\Eu{H}$.
But this holds by Lemma \ref{invariance:next:lemma}.
\end{proof}

\section{Lazy AC-rewriting with strategies}
\label{lazy:rewriting:sec}

In this section we integrate lazy AC-matching with strategy application.   
More details on strategy languages can be found in 
\cite{BKKR-IJFCS-2001,Vis01-rta,Marti-OlietMV05}.

Primitive strategies are rewrite rules $\rrule{l}{r}$ and the
\textsf{id} and \textsf{fail} strategies that respectively always and
never succeed. They are completed with the most
usual reduction strategies, namely the four traversal strategies
\textit{leftmost-outermost}, \textit{leftmost-innermost},
\textit{parallel-outermost} and \textit{parallel-innermost}~\cite[Definition
4.9.5]{Terese03} that control a rewrite system by selecting redexes according to
their position. For sake of simplicity we restrict their control to a single
rewrite rule. Let $v$ be one of these four strategies. The application of $v$ to
the rewrite rule $\rrule{l}{r}$ is denoted by $v(\rrule{l}{r})$. 
The sequential composition of two strategies $u$ and $w$ is denoted by $u;w$.
The application of a strategy $u$ to a term $t$ is denoted by $[u] \cdot t$.

A strategy application produces a lazy list of terms, defined by the grammar
$\Eu{L}~::=~\bot_{\EuScript{T}}$ $|$ $\Eu{L}~::~\Eu{L}$ $|$
$\Eu{T}$ $|$ $\Eu{C}(\Eu{T})$.
 A list is usually defined by a constructor for an empty list and a constructor
adding one element at the head of another list. Then concatenation of two
lists is defined, with another notation. Here we equivalently introduce an
associative symbol $::$ for concatenation of two lists of terms, and a symbol
$\bot_{\EuScript{T}}$ to denote an empty list of terms, which is a neutral
element for $::$. We use the same conventions as for $\lor$ in
Section~\ref{ac:match:sec}. Let $\llist{\EuScript{T}}$ denote the set of lazy
lists of terms.

\begin{figure}[bht!]
\setlength{\abovecaptionskip}{-0.3cm}

\subfigure[\textsf{id}, \textsf{fail} and composition rules]{ 
$
\begin{array}{|r  l|}
\hline
\texttt{identity:} & [\textsf{id}]\cdot \tau \leadsto \tau \\
\texttt{failure:}& [\textsf{fail}]\cdot \tau  \leadsto \bot_{\EuScript{T}}  \\
\texttt{compose:} & [u;v]\cdot \tau \leadsto [u] \cdot ([v] \cdot \tau) \\
\hline
\end{array}
$
 \label{primitive:tab}
}
$\quad$
\subfigure[Top rewriting]{
$
\begin{array}{@{}|r@{\;}l|@{}}
\hline
\texttt{rule1:} &
 [\rrule{l}{r}]\cdot \bot_{\EuScript{T}} \leadsto
   \bot_{\EuScript{T}} \\
\texttt{rule2:} &[\rrule{l}{r}]\cdot (t :: {\tau}) \leadsto 
  (\matchth{l}{t}{AC})(r) :: ([\rrule{l}{r}]\cdot \tau) \\
\texttt{subs\_fail:}& \textbf{F} (t) \leadsto \bot_{\EuScript{T}} \\
\texttt{subs\_id:}& \textbf{I} (t) \leadsto t \\
\texttt{subs:}& (\sigma\lor C)(t) \leadsto \sigma(t) :: {C(t)} \\
\hline
\end{array}
$
 \label{AC-rewriting:top}
}
\caption{AC-rewriting operational semantics}
\end{figure}

The operational semantics of the strategies \textsf{id} and \textsf{fail} and of
strategy composition are defined in Figure~\ref{primitive:tab} for any lazy list
of terms $\tau$ and any two strategies $u$ and $v$. The operational semantics of
top rewriting is defined in Figure~\ref{AC-rewriting:top}. Let \texttt{LTR} be
the system composed of these five rules and the \lazy AC-matching. The rules
\texttt{rule1} and \texttt{rule2} reduce the application of a rewrite rule at the top of terms in a lazy list of
terms. In \texttt{rule2} the expression $\matchth{l}{t}{AC}$ is reduced to its
normal form by \lazy. The result is a lazy list of constraints.
 The rules
\texttt{subs\_fail}, \texttt{subs\_id} and \texttt{subs} reduce the application
of a lazy list of constraints on a term. In the right side of rule
\texttt{subs},  a $\land$-substitution $\sigma$ is applied to a term, in a
sense which is a simple extension of the standard definition of substitution
application. Equivalently $\land$-substitutions can be reduced to standard ones
by a transformation similar to the post-processing  defined for the eager
AC-matching.

 The application of the system
\texttt{LTR} to the expression $[\rrule{l}{r}]\cdot(t)$ produces either
$\bot_{\EuScript{T}}$ or a non-empty lazy list of terms $u_1 :: \tau_1$,
where $u_1$ is the first result of the application of the rewrite rule $\rrule{l}{r}$ at the top of the
term $t$ and $\tau_1$ is (a syntactic object denoting) the lazy list of the
other results. When applying on $\tau_1$ the rewrite system defined in
Figure~\ref{Next:term},
 and then the system \texttt{LTR}, we get again
either $\bot_{\EuScript{T}}$ or a non-empty lazy list $u_2 :: \tau_2$, where $u_2$ is
the second result of the application of the rewrite rule $\rrule{l}{r}$ at the top of the
term $t$ and $\tau_2$ is again (a syntactic object denoting) a lazy list of
terms that represents the remaining results, and so on. 

\begin{figure}[hbt!]
\begin{align*}
\begin{array}{|l l l l|}
\hline
\texttt{next\_empty:}&  \textit{Next}(\bot_{\EuScript{T}}) \leadsto \bot_{\EuScript{T}}
&\hspace{1.5cm} \texttt{next\_app:}&  \textit{Next}(L_1 :: L_2) \leadsto 
  \textit{Next}(L_1) :: L_2  \\
\texttt{next\_term:}&  \textit{Next}(t) \leadsto t,  \;\; \textrm{if } t \textrm{
is a term} 
&\hspace{1.5cm} \texttt{next\_cstr:}& \textit{Next}(C(t)) \leadsto 
 \textit{Next}(C)(t) \\
\hline 
\end{array}
\end{align*}
\caption{\textit{Next} rules for lists of terms\label{Next:term}}
\end{figure}

The application of a traversal strategy on a lazy list of terms in
$\llist{\EuScript{T}}$ is defined as follows:
\begin{align*}
\begin{array}{|r  l|}
\hline
 \texttt{traversal1:}
    & [v(\rrule{l}{r})]\cdot \bot_{\EuScript{T}} \leadsto
       \bot_{\EuScript{T}} \\
 \texttt{traversal2:}
    & [v(\rrule{l}{r})]\cdot(t :: \tau) \leadsto 
       [v(\rrule{l}{r})]\cdot t :: [v(\rrule{l}{r})]\cdot \tau\\
\hline       
\end{array}       
\end{align*}
 The rules
\begin{align*}
[v(u)]\cdot t\leadsto 
\begin{cases} 
[u]\cdot t 
  & \textrm{ if } [u]\cdot t \neq \bot_{\EuScript{T}} 
     \textrm{ or } t \in \varset     \\ 
  \up f([v(u)]\cdot t_1,\ldots,[v(u)]\cdot t_n)
  & \textrm{ if } [u]\cdot t = \bot_{\EuScript{T}} 
        \textrm{ and }  t=f(t_1,\ldots,t_n)   
\end{cases}
\end{align*}
define the application of the traversal strategy $v(u)$ on the term $t$ for the
rewrite rule $u$ and the \textit{parallel-outermost} strategy constructor $v$. 
The other traversal strategies can be handled similarly. We have seen that the
application of a rewrite rule at the top of a term yields a lazy list of terms in
$\llist{\EuScript{T}}$. Here the application of a rule to a term at arbitrary
depth, via a traversal strategy, yields a \emph{decorated} term, which is a term
where some subterms are replaced by a lazy list of terms. This lazy list of
terms is abusively called a \emph{lazy subterm}, with the property that lazy
subterms are not nested. In other words, the positions of two lazy subterms are not comparable, for the
standard prefix partial order over the set of term positions. The operator $\up$
is assumed to reduce a decorated term to a lazy list of terms. We
summarize its behavior as follows. A decorated term can be encoded by a
tuple $(t,k,p,\delta)$ where $t$ is the term before strategy application, $k$ is the number of
decorated positions, $p$ is a function from $\{1, \ldots, k\}$ to the domain of
$t$ (i.e. its set of positions) which defines the decorated positions, such that
$p(i)$ and $p(j)$ are not comparable if $i \neq j$, and $\delta$ is the function
from $\{1, \ldots, k\}$ to $\llist{\EuScript{T}}$ such that $\delta(i)$ is the
lazy list at position $p(i)$, $1 \leq i \leq k$. It is easy to construct an
iterator over the $k$-tuples of positive integers (for instance in
lexicographical order), and to derive from it an iterator over tuples of terms
$(s_1, \ldots, s_k)$ with $s_i$ in the list $\delta(i)$ for $1 \leq i \leq k$.
From this iterator and function $p$ we derive an iterator producing the lazy
list $\up t$ by replacing each subterm $t_{|p(i)}$ by $s_i$.

\section{Implementation and experiments}
\label{implem:sec}
We present here a prototypal implementation of lazy AC-matching and report about
its experimentation. Our implementation  is  a straightforward
translation of the \lazy system in the rule-based language
\texttt{symbtrans} \cite{BGL-JSC10} built on the computer algebra system Maple.

This prototype obviously does not claim
 efficiency in the usual sense of the number of solutions computed in a given
amount of time. But this section shows that our prototype optimises the standard
deviation of the time between two successive solutions. This performance
criterion corresponds to our initial motivations and can be measured on the
prototype. We consider the matching problem
$\matchAC{x_1+\ldots+x_{18}}{a_1+\ldots+a_{18}}$, for 18 variables $x_1$,
\ldots, $x_{18}$ and 18 constants $a_1$, \ldots, $a_{18}$. On this problem our
lazy prototype provides any two consecutive solutions in an average time of
$0.37$ seconds, with a standard deviation of $0.021$ seconds between the 100-th
first solutions. In comparison the computation time between two consecutive
solutions with the Maude function \texttt{metaXmatch} grows exponentially.
The experiment is done on an Intel core 2 Duo T6600@2.2GHz with 3.4Gb of
memory, under a x86\_64 Ubuntu Linux.

Finally, it is worth mentioning the performance of the Maple standard matching
procedure \texttt{pat\-match(expr, pattern,'s')} that returns true if it is able
to match \texttt{expr} to \texttt{pattern}, and false otherwise. If the matching
is successful, then \texttt{s} is assigned a substitution such that $AC
\models\texttt{s(pattern)=expr}$.
This procedure runs out of memory if the arity of the AC symbols is large.
With Maple 14 this failure can be observed  when computing a solution of the
matching problem $\matchAC{x_1+\ldots+x_{12}}{a_1+\ldots+a_{12}}$.

\section{Conclusion}
\label{conclusion:sec}
We presented a lazy AC-matching algorithm and a lazy evaluation semantics for
AC-rewriting and some basic strategies. The semantics is designed to be
implemented in a strict language by representing delayed matching
constraints and lazy lists of terms by explicit objects. We also described a
common principle for lazy traversal strategies. The potential benefits are
clear: performances are dramatically increased by avoiding unnecessary
computations. We are working on an implementation of lazy AC-matching and
AC-rewriting: first results show that our approach is  efficient when the arity
of AC symbols is high and when the number of solutions of the AC-matching
problem is large. However, we do not claim efficiency for the search of the
first solution by the AC-matching algorithm.

Here no neutral element is assumed for AC symbols. As a consequence the lazy
AC-matching relies on a surjection iterator. We plan to address the question of
its efficiency, and to extend the present work to AC symbols with a neutral
element. Our intuition is that our approach can easily be adapted to that
case.

\bibliographystyle{eptcs}

\begin{thebibliography}{10}
\providecommand{\bibitemdeclare}[2]{}
\providecommand{\urlprefix}{Available at }
\providecommand{\url}[1]{\texttt{#1}}
\providecommand{\href}[2]{\texttt{#2}}
\providecommand{\urlalt}[2]{\href{#1}{#2}}
\providecommand{\doi}[1]{doi:\urlalt{http://dx.doi.org/#1}{#1}}
\providecommand{\bibinfo}[2]{#2}

\bibitemdeclare{unpublished}{BGL-JSC10}
\bibitem{BGL-JSC10}
\bibinfo{author}{W.~Belkhir}, \bibinfo{author}{A.~Giorgetti} \&
  \bibinfo{author}{M.~Lenczner} (\bibinfo{year}{December 2010}):
  \emph{\bibinfo{title}{Rewriting and Symbolic Transformations for Multi-scale
  Methods}}.
\newblock \bibinfo{note}{Url: http://arxiv.org/abs/1101.3218v1. Submitted.}

\bibitemdeclare{inproceedings}{Benanav:1985:CMP}
\bibitem{Benanav:1985:CMP}
\bibinfo{author}{D.~Benanav}, \bibinfo{author}{D.~Kapur} \&
  \bibinfo{author}{P.~Narendran} (\bibinfo{year}{1985}):
  \emph{\bibinfo{title}{Complexity of matching problems}}.
\newblock In: {\sl \bibinfo{booktitle}{Proc. of the 1st Int. Conf. on Rewriting
  Techniques and Applications}}, {\sl \bibinfo{series}{LNCS}}
  \bibinfo{volume}{202}, \bibinfo{publisher}{Springer}, pp.
  \bibinfo{pages}{417--429}.
 \doi{10.1007/3-540-15976-2\_22}.  

\bibitemdeclare{article}{BKKR-IJFCS-2001}
\bibitem{BKKR-IJFCS-2001}
\bibinfo{author}{P.~Borovansk\'y}, \bibinfo{author}{C.~Kirchner},
  \bibinfo{author}{H.~Kirchner} \& \bibinfo{author}{C.~Ringeissen}
  (\bibinfo{year}{2001}): \emph{\bibinfo{title}{Rewriting with strategies in
  {{\sf ELAN}}: a functional semantics}}.
\newblock {\sl \bibinfo{journal}{{International Journal of Foundations of
  Computer Science}}} \bibinfo{volume}{12}(\bibinfo{number}{1}), pp.
  \bibinfo{pages}{69--98}. \doi{10.1142/S0129054101000412}.

\bibitemdeclare{article}{CFK07}
\bibitem{CFK07}
\bibinfo{author}{H.~Cirstea}, \bibinfo{author}{G.~Faure} \&
  \bibinfo{author}{C.~Kirchner} (\bibinfo{year}{2007}): \emph{\bibinfo{title}{A
  $\rho$-calculus of explicit constraint application}}.
\newblock {\sl \bibinfo{journal}{Higher-Order and Symbolic Computation}}
  \bibinfo{volume}{20}, pp. \bibinfo{pages}{37--72}.
  \doi{10.1007/s10990-007-9004-2}.

\bibitemdeclare{article}{rhoCalIGLP-I+II-2001}
\bibitem{rhoCalIGLP-I+II-2001}
\bibinfo{author}{H.~Cirstea} \& \bibinfo{author}{C.~Kirchner}
  (\bibinfo{year}{2001}): \emph{\bibinfo{title}{The rewriting calculus ---
  {Part~I {\em and} II}}}.
\newblock {\sl \bibinfo{journal}{Logic Journal of the Interest Group in Pure
  and Applied Logics}} \bibinfo{volume}{9}(\bibinfo{number}{3}), pp.
  \bibinfo{pages}{427--498}.

\bibitemdeclare{proceedings}{MaudeBook07}
\bibitem{MaudeBook07}
\bibinfo{editor}{M.~Clavel}, \bibinfo{editor}{F.~Dur{\'a}n},
  \bibinfo{editor}{S.~Eker}, \bibinfo{editor}{P.~Lincoln},
  \bibinfo{editor}{N.~Mart\'{\i}-Oliet}, \bibinfo{editor}{J.~Meseguer} \&
  \bibinfo{editor}{C.~L. Talcott}, editors (\bibinfo{year}{2007}):
  \emph{\bibinfo{title}{All About Maude - A High-Performance Logical Framework,
  How to Specify, Program and Verify Systems in Rewriting Logic}}. {\sl
  \bibinfo{series}{LNCS}} \bibinfo{volume}{4350},
  \bibinfo{publisher}{Springer}.

\bibitemdeclare{article}{Dersh_Ordering82}
\bibitem{Dersh_Ordering82}
\bibinfo{author}{N.~Dershowitz} (\bibinfo{year}{1982}):
  \emph{\bibinfo{title}{Ordering for Term-Rewriting Systems}}.
\newblock {\sl \bibinfo{journal}{Theoretical Computer Science}}
  \bibinfo{volume}{17}, pp. \bibinfo{pages}{279--300}.
  \doi{10.1016/0304-3975(82)90026-3}.

\bibitemdeclare{article}{Bipartie-Eker95}
\bibitem{Bipartie-Eker95}
\bibinfo{author}{S.~Eker} (\bibinfo{year}{1995}):
  \emph{\bibinfo{title}{AC-Matching Via Bipartite Graph Matching}}.
\newblock {\sl \bibinfo{journal}{Comput. J.}}
  \bibinfo{volume}{38}(\bibinfo{number}{5}), pp. \bibinfo{pages}{381--399}.
\doi{10.1093/comjnl/38.5.381}.

\bibitemdeclare{article}{Fast-Eker96}
\bibitem{Fast-Eker96}
\bibinfo{author}{S.~Eker} (\bibinfo{year}{1996}): \emph{\bibinfo{title}{Fast
  matching in combinations of regular equational theories}}.
\newblock {\sl \bibinfo{journal}{ENTCS}} \bibinfo{volume}{4}, pp.
  \bibinfo{pages}{90--109}.

\bibitemdeclare{article}{Single-Eker02}
\bibitem{Single-Eker02}
\bibinfo{author}{S.~Eker} (\bibinfo{year}{2002}): \emph{\bibinfo{title}{Single
  Elementary AC-Matching}}.
\newblock {\sl \bibinfo{journal}{J. Autom. Reasoning}}
  \bibinfo{volume}{28}(\bibinfo{number}{1}), pp. \bibinfo{pages}{35--51}.
 \doi{10.1023/A:1020122610698}.
 
\bibitemdeclare{inproceedings}{gramlich-unif88}
\bibitem{gramlich-unif88}
\bibinfo{author}{B.~Gramlich} (\bibinfo{year}{1988}):
  \emph{\bibinfo{title}{Efficient {AC}-Matching using Constraint Propagation}}.
\newblock In: {\sl \bibinfo{booktitle}{Proc.\ 2nd Int.\ Workshop on
  Unification, Internal Report 89 R 38, CRIN}}, \bibinfo{address}{Val d'Ajol,
  France}.

\bibitemdeclare{article}{Kirchner:2001:Promoting:AC}
\bibitem{Kirchner:2001:Promoting:AC}
\bibinfo{author}{H.~Kirchner} \& \bibinfo{author}{P.-E. Moreau}
  (\bibinfo{year}{2001}): \emph{\bibinfo{title}{Promoting rewriting to a
  programming language: a compiler for non-deterministic rewrite programs in
  {AC}-theories}}.
\newblock {\sl \bibinfo{journal}{J. Funct. Program.}} \bibinfo{volume}{11}, pp.
  \bibinfo{pages}{207--251}.

\bibitemdeclare{article}{Marti-OlietMV05}
\bibitem{Marti-OlietMV05}
\bibinfo{author}{N.~Mart\'{\i}-Oliet}, \bibinfo{author}{J.~Meseguer} \&
  \bibinfo{author}{A.~Verdejo} (\bibinfo{year}{2005}):
  \emph{\bibinfo{title}{Towards a Strategy Language for {Maude}}}.
\newblock {\sl \bibinfo{journal}{Electr. Notes Theor. Comput. Sci.}}
  \bibinfo{volume}{117}, pp. \bibinfo{pages}{417--441}.
\doi{10.1016/j.entcs.2004.06.020}.

\bibitemdeclare{book}{Terese03}
\bibitem{Terese03}
\bibinfo{author}{Terese} (\bibinfo{year}{2003}): \emph{\bibinfo{title}{Term
  Rewriting Systems}}.
\newblock {\sl \bibinfo{series}{Cambridge Tracts in Theor. Comp.
  Sci.}}~\bibinfo{volume}{55}, \bibinfo{publisher}{Cambridge Univ. Press}.

\bibitemdeclare{inproceedings}{Vis01-rta}
\bibitem{Vis01-rta}
\bibinfo{author}{E.~Visser} (\bibinfo{year}{2001}):
  \emph{\bibinfo{title}{Stratego: {A} Language for Program Transformation based
  on Rewriting Strategies. {S}ystem Description of {Stratego} 0.5}}.
\newblock In: {\sl \bibinfo{booktitle}{Proc. of RTA'01}}, {\sl
  \bibinfo{series}{Lecture Notes in Computer Science}} \bibinfo{volume}{2051},
  \bibinfo{publisher}{Springer-Verlag}, pp. \bibinfo{pages}{357--361}.
 \doi{10.1007/3-540-45127-7\_27}.
 
\bibitemdeclare{inproceedings}{EuroSim11}
\bibitem{EuroSim11}
\bibinfo{author}{B.~Yang}, \bibinfo{author}{W.~Belkhir}, \bibinfo{author}{R.N.
  Dhara}, \bibinfo{author}{M.~Lenczner} \& \bibinfo{author}{A.~Giorgetti}
  (\bibinfo{year}{2011}): \emph{\bibinfo{title}{Computer--Aided Multiscale
  Model Derivation for {MEMS} Arrays}}.
\newblock In: {\sl \bibinfo{booktitle}{{EuroSimE 2011, 13-th Int. Conf. on
  Thermal, Mechanical and Multi-Physics Simulation and Experiments in
  Microelectronics and Microsystems}}}, \bibinfo{publisher}{IEEE Computer
  Society}, \bibinfo{address}{Linz, Austria}.
\newblock \bibinfo{note}{6 pages}. \doi{10.1109/ESIME.2011.5765784}.

\bibitemdeclare{inproceedings}{CFM11}
\bibitem{CFM11}
\bibinfo{author}{B.~Yang}, \bibinfo{author}{R.N. Dhara},
  \bibinfo{author}{W.~Belkhir}, \bibinfo{author}{M.~Lenczner} \&
  \bibinfo{author}{A.~Giorgetti} (\bibinfo{year}{2011}):
  \emph{\bibinfo{title}{Formal Methods for Multiscale Models Derivation}}.
\newblock In: {\sl \bibinfo{booktitle}{CFM 2011, 20th Congr\`es Fran\c{c}ais de
  M\'ecanique}}.
\newblock \bibinfo{note}{5 pages}.

\end{thebibliography}

\end{document}